\documentclass[
reprint,            
nofootinbib,        
superscriptaddress  
]{revtex4-1}        

\usepackage[utf8]{inputenc}
\usepackage[T1]{fontenc}
\usepackage[portuguese]{babel}
\usepackage{amsmath} 
\usepackage{amssymb} 
\usepackage{amsthm}  
\usepackage{hyperref}  
\usepackage{tikz}      
\usepackage{tcolorbox} 
\usepackage{dcolumn}   
\usepackage{graphicx, wrapfig, setspace, booktabs, float, epsfig} 
\usepackage{multirow}  
\usepackage{enumitem}
\usepackage{tikz}

\theoremstyle{plain}
\newtheorem{theorem}{Theorem}[section]
\newtheorem*{theorem*}{Theorem}

\newtheorem{corollary}{Corollary}

\begin{document}

\title{Constraints Between Equations of State and mass-radius relations in General Clusters of Stellar Systems}
\author{Y. X. Martins}
\affiliation{Departamento de Matem\'atica, ICEx, Universidade Federal de Minas
    Gerais,    Av. Ant\^onio Carlos 6627, Pampulha, CP 702, CEP 31270-901, Belo Horizonte, MG, Brazil}
\author{L. F. A. Campos}
\affiliation{Departamento de F\'isica, ICEx, Universidade Federal de Minas
    Gerais, Av. Ant\^onio Carlos 6627, Pampulha, CP 702, CEP 31270-901, Belo Horizonte, MG, Brazil}
\author{D. S. P. Teixeira}
\affiliation{Departamento de F\'isica, ICEx, Universidade Federal de Minas
    Gerais, Av. Ant\^onio Carlos 6627, Pampulha, CP 702, CEP 31270-901, Belo Horizonte, MG, Brazil}
\author{R. J. Biezuner}
\affiliation{Departamento de Matem\'atica, ICEx, Universidade Federal de Minas
    Gerais,    Av. Ant\^onio Carlos 6627, Pampulha, CP 702, CEP 31270-901, Belo Horizonte, MG, Brazil}

\begin{abstract}
We prove three obstruction results on the existence of equations of state in clusters of stellar systems fulfilling mass-radius relations and some additional bound (on the mass, on the radius or a causal bound). The theorems are proved in great generality. We start with a motivating example of TOV systems and apply our results to stellar systems arising from experimental data. 
\end{abstract}

\pacs{Valid PACS appear here}
\maketitle

\section{Introduction}

$\quad\;\,$The structure of general relativistic and spherically
symmetric isotropic stars is modeled by the Tolman-Oppenheimer-Volkoff (TOV) equations \cite{weinberg}, described in terms of its density $\rho$ and pressure $p$ (setting $G=1$ and $c=1$):
\begin{eqnarray}
\label{tov1}
p'(r) & =- & \frac{\Bigl(\rho(r)+p\Bigr)\Bigl(M(r)+4\pi r^{3}p\Bigr)}{r^{2}\bigl(1-\frac{2M(r)}{r}\bigr)}\\
\label{tov2} 
M'(r) & = & 4\pi r^{2}\rho(r).
\end{eqnarray}
These equations are partially uncoupled and can be coupled through an equation of state. A typical example is the polytropic equation of state
\begin{equation}\label{politropic}
f(p,\rho,k,k_0,\gamma) = p - k\rho^\gamma - k_0 =0,
\end{equation}
where $k, k_0 \in\mathbb{R}$  are the \emph{polytropic constant} and the \emph{stiffness constant}, and $\gamma=(n+1)/n \in \mathbb{Q}$ is the \emph{polytrope exponent}. Since $f(p,\rho) = 0$ can obviously be solved for $p$ and $\rho$, it can be used to couple a TOV system. Under this coupling, we say that $(p,\rho)$ corresponds to a \emph{polytropic TOV system}.

Now, we recall that in thermodynamic systems where pressure and density are related by a solvable equation of state, the speed of sound within the system can be defined by $v :=\sqrt{ \partial p/\partial \rho}$. It folows from (\ref{politropic}) that $v=\sqrt{k\gamma p^{\gamma-1}}$ in polytropic TOV systems. A TOV system is \emph{causal} if the speed of sound does not exceeds the speed of light, that is, $k\gamma\rho^{\gamma - 1} < 1$. In particular, 
\begin{equation}\label{causalcentro}
k\gamma\rho_c^{\gamma - 1}<1,
\end{equation}
where $\rho_c=\rho(0)$ is the density at the origin.

Suppose a star of radius $R$ is modeled by a TOV system, and let the \emph{mass of the system} to be  $M=M(R)$, where $M(r)$ is the function in (\ref{tov2}). Small stars usually admit \emph{mass-radius relations} equations of state, $g(M,R,\delta)=0$, relating $M$ to $R$ and other parameters \cite{stellar_structure}. For instance, it is known  \cite{weinberg,stellar_structure} that, in the Newtonian limit, polytropic TOV systems satisfy
\begin{equation}\label{polimassaraio}
g(M,R,k,n,\rho_c)=M- A(k,n,\rho_c)R^2,
\end{equation}
where
\begin{equation*}
A(k,n,\rho_{c})=\Bigl(\frac{4\pi}{(n+1)k}\Bigr)^{3/2}\rho_{c}^{(n-3)/2n}\frac{\vert p'(R)\vert}{\rho_{\mathrm{rel}}}
\end{equation*}
and $\rho_{\mathrm{rel}}=\rho(R)/\rho_{c}$ is the \emph{relative density}.

We show that there are obstructions for a Newtonian polytropic TOV system to be causal. 

This follows from isolating $\rho_c$ can be isolated in (\ref{polimassaraio}) and substituting in (\ref{causalcentro}), leading to the following obstruction on the possible polytropic constants 
\begin{equation*}
k<\frac{n}{(n+1)}\Bigl(\frac{R^{4}}{n^{3}M^{2}}\beta\Bigr)^{\frac{1}{n}},\;\text{where}\;\beta=64\pi^{3}\Bigl(\frac{\vert p'(R)\vert}{\rho_{\mathrm{rel}}}\Bigr)^{2}.
\end{equation*}

In this article, we show that this obstruction is not particular to polytropic Newtonian TOV systems. We define precisely \emph{clusters} of stellar systems (of which TOV systems particular examples), equations of state and mass-radius relations, in in Section \ref{definitions}. Then, in Section \ref{statement_proof}, we state and prove three different obstruction theorems which formalize the rough claim below. Finally, in Section \ref{applications}, we apply our results to stellar systems arising from experimental data.
\begin{theorem*}[Roughly]\label{thm_obstructions}
Consider a small radius stellar system with a mass-radius relation $g(M,R,\epsilon)$. Each constraint on $\epsilon$ induces an obstruction on the possible equations of state depending on $\epsilon$ that can be introduced in that system.
\end{theorem*}
Before giving the precise statement and the proof, let us notice that this claim is reasonable. Generally, mass-radius relations are closely related to the atomic nature of the star, as stellar systems with different atomic constitution obey different relations \cite{stellar_structure,RMR_polinomial_3}. Reciprocally, constraints on the parameters of a mass-radius relation provide information about the atomic structure of the system. For instance, Newtonian polytropic stars satisfying (\ref{polimassaraio}), but not the Chandrasekhar limit (resp. Oppenheimer-Volkoff limit) cannot be stable white dwarfs (resp. stable neutron stars); there are also bounds on $n$, of course \cite{weinberg}. On the other hand, equations of state arise from the statistic mechanics of the atomic structure. Summarizing, mass-radius relations with constraints restrict the atomic structure and, therefore, the possible equations of state, which is precisely the content of Theorem \ref{thm_obstructions}.

\section{Definitions}\label{definitions}

A \emph{stellar system} is defined as a pair $(p,\rho)$ of piecewise differentiable functions $I\rightarrow \mathbb{R}$. It is natural to define a \emph{cluster of stellar systems} of degree $(k,l)$ as a vector subspace $\mathrm{Stellar}^{kl}(I)$ of $C_{pw}^{k}(I)\times C_{pw}^{l}(I)$. It is important to notice that $C_{pw}^{k}(I)$  has a canonical generalized norm, which is a map satisfying the norm axioms but possibly taking values at infinity. In fact, if $f\in C_{pw}^{k}(I)$, $\Vert f\Vert_{k}:=\sup_{I}\vert f^{(k)}(t)\vert$ is such a generalized norm. Like classical norms, generalized norms induce a topology which \cite{generalized_norm} can be characterized as the finest locally convex topology that makes sum and scalar multiplication continuous. Therefore, $\mathrm{Stellar}^{kl}(I)$ has a canonical locally convex space structure. 

We are interested in stellar systems having well defined notions of mass and of radius. This leads us to consider clusters of stellar systems that becomes endowed with maps $M,R:\mathrm{Stellar}^{kl}(I)\rightarrow\mathbb{R}$ assigning to each system $(p,\rho)$ its mass and its radius. In TOV systems, $M(p,\rho)$ is given by (\ref{tov2}), and its inverse is $R(p,\rho)$. Looking at these expressions, we see that it is natural to assume $M$ and $R$ at least piecewise continuous. In other words, $M,R\in C_{pw}^{0}(\mathrm{Stellar}^{kl}(I))$. We say that a (locally convex) subspace $\mathrm{Bound}^{kl}(I)\subset\mathrm{Stellar}^{kl}(I)$ has \emph{mass bounded from above} (resp. \emph{radius bounded from above}) if when restricted to it the function $M$ (resp. $R$) is bounded from above. Similarly, we could define subspaces with mass and radius bounded from below. 

A \emph{state function} for a cluster of stellar systems is a function $f:\mathrm{Stellar}^{kl}(I)\times E\rightarrow\mathbb{R}$, where $E$ is a topological vector space of parameters. We say that a state function is \emph{locally integrable at $p$} if the corresponding equation of state $f(p,\rho,\epsilon)=0$ can be locally solved for $p$. In other words, if there is a neighborhood $U$ of $p$, open sets $V\subset C_{pw}^{l}(I)$ and $W\subset E$, and a function $\xi:V\times W\rightarrow U$ such that $f(\xi(\rho,\epsilon),\rho,\epsilon)=0$. additionally, if $\xi$ is monotone in both variables we say that $f$ is a \emph{monotonically locally integrable (MLI)} state function.

When a locally integrable state function $f$ is such that $\xi$ is differentiable (in some sense to be specified below), we can define the \emph{squared speed of sound} as $v^2=\partial_{\rho}\xi$. We will be interested in situations in which both $\xi$ and $v$ are monotone. A state function satisfying these properties is \emph{fully monotonically locally integrable}.  

The notion of differentiability of $f$ (and, therefore, of derivative) is linked to the nature of the space of parameters $E$, leading to different generalizations of the Implicit Function Theorem (IFT). These versions of the IFT imply that any function whose derivative in the direction of $p$ satisfies mild conditions is locally integrable. 

We remark that some classical generalizations of IFT, such as the IFT for Banach spaces \cite{IFT_cartan} and the Nash-Moser Theorem \cite{Nash_Moser} for Fréchet spaces, cannot be used here, because $\mathrm{Stellar}^{kl}(I)$ isn't either Banach or  Fréchet. General contexts that apply here are when $E$ is an arbitrary topological vector space and, more concretely, when $E$ is locally convex (see \cite{IFT_geral} and \cite{IFT_loc_convex}, respectively). We will work in an intermediary context: when $E$ is Banach, so that by embedding $\mathbb{R}$ in $E$ we can regard $f$ as a map $f:\mathrm{Stellar}^{kl}(I)\times E\rightarrow E$ and define the directional derivative $\partial_{p}f$ as usual for maps between Banach spaces, without the technicalities needed in the general situations. As proved in \cite{IFT_adequado}, the classical IFT holds, meaning that if $\partial_{p}f\neq0$, then $f$ is a locally integrable state function in a neighborhood of that point and $\xi$ is differentiable (in a certain generalized sense). This IFT also gives an expression for the derivative of $\xi$ at each point which depends on the derivative of $f$ with respect to the other variables at points of the neighborhood where $\xi$ is defined, and $\xi$ is monotone if $\partial _i f \neq 0$ in those points.

A \emph{mass-radius function} for a cluster of stellar systems is a function $g:C_{pw}^{0}(\mathrm{Stellar}^{kl}(I))\times C_{pw}^{0}(\mathrm{Stellar}^{kl}(I))\times F\rightarrow\mathbb{R}$, where $F$ is a topological vector space of parameters, such that the mass-radius relation $g(M,R,\delta)=0$ can be locally solved for some $\delta$, in that we can locally write $\delta=\eta(M,R)$. If $\eta$ is monotone, we say that $g$ is a \emph{locally monotone mass-radius function}. As in the previous case, if  $F$ is Banach, then a mass-radius function can be obtained by requiring $\partial_{\delta}g\neq0$, and locally monotone if the other partial derivatives also do not vanish. 

\section{Statement and Proof} \label{statement_proof}

\begin{theorem}\label{bound}
Let  $\mathcal{C}=(\mathrm{Stellar}^{kl}(I),M,R)$  be a cluster of stellar systems of degree $(k,l)$ endowed with a locally monotone mass-radius function $g(M,R,\epsilon)$. Then any upper (resp. lower) bound on the mass and on the radius induces a bound on each MLI state function of $\mathcal{C}$ depending on $\epsilon$ (and possibly on other parameters). If a state function is fully MLI at some $p$, then the bounds descend to the speed of sound near $p$.
\end{theorem}
\begin{proof}
After the previous discussion, the proof becomes easy. We  only work with upper bounds; the proof for lower bounds is essentially the same. By definition, the mass-radius relation $g(M,R,\epsilon)=0$ can be solved in a neighborhood of some $\epsilon$, that is, locally $\epsilon=\eta(M,R)$.  Since $M$ and $R$ are bounded, say $M\leq m$ and $R\leq r$, and $g$ is locally monotone, we see that $\eta(M,R)\leq\eta(m,R)\leq\eta(m,r)$. Therefore, $\epsilon$ is bounded too, say by $\epsilon_{0}$. Now, notice that if a parameter $\epsilon$ is bounded then any monotone function $f(x,\epsilon)$ depending on that parameter is bounded by the function $h(x)=f(x,\epsilon_{0})$. In particular, any MLI function of state $f(p,\rho,\epsilon)$ for the cluster $\mathcal{C}$ depending on $\epsilon$, has an upper bound. Furthermore, if $f$ is fully MLI at some $p$, then the speed of sound squared $v(\rho,\epsilon)$ near $p$ is well defined and it also depends monotonically of $\epsilon$, so that it also has an upper bound.
\end{proof}

The above theorem only applies for clusters whose function of state and mass-radius function depend on the same parameter. This hypothesis can be avoided by adding some continuity equation. Let $\mathcal{C}=(\mathrm{Stellar}^{kl}(I),M,R)$ be a cluster of stellar systems in which $M$ is at least piecewise $C^1$. A \emph{continuity equation} for $\mathcal{C}$ is an ordinary differential equation $M'(R)=F(R,\rho)$. We will  only work with continuity equations which can be locally solved for $\rho$. The basic example is equation (\ref{tov2}).

\begin{theorem}\label{bound2}
Let  $\mathcal{C}=(\mathrm{Stellar}^{kl}(I),M,R)$  be a cluster of stellar systems of degree $(k,l)$ endowed with a locally monotone mass-radius function $g(M,R,\delta)$ and a continuity equation. Then any upper (resp. lower) bound on the derivative of the mass and on the radius induces a bound on each MLI state function of $\mathcal{C}$. If a state function is fully MLI at some $p$, then bounds are induced on the speed of sound near $p$.
\end{theorem}
\begin{proof}
The proof is similar. Because the mass-radius function $g(M,R,\delta)$ is locally monotone, it can be locally solved for $M$, allowing us to write $M(R,\delta)$. From the continuity equation, we have $\partial _R M(R,\delta)=F(R,\rho)$. On the other hand, the RHS can be solved for $\rho$ as $\rho (R)=\partial _R M(R,\delta)$. It follows that  any bound $\partial _R M(R,\delta)\leq M_0 (\delta)$ induces a bound  $\rho (R)\leq M_0 (\delta)$. Consequently, if $f(p,\rho,\epsilon)$ is a MLI function of state, we have the desired bound $f(p,\rho,\epsilon)\leq M_0 (\delta)$, which clearly descends to the speed of sound when $f$ is fully MLI.
\end{proof}

In the last two theorems we showed that bounds on the mass-radius relations induce bounds on the functions of state, which descend to the speed of sound. On the other hand, additional conditions could be imposed \emph{a priori} on the speed of sound, such as causal conditions. We show that these conditions lift to new bounds on the state functions. In order to do this, given a fully MLI function of state $f(p,\rho,\delta)$, let us define a \emph{causal condition} for $f$ near $p$ as an upper bound $v^2(\rho,\epsilon)\leq v_0^2(\rho,\epsilon)$ for the speed of sound near $p$. We assume that both $v^2$ and $v_0^2$ are \emph{locally invertible}, meaning that we can locally invert $v^2(\rho,\epsilon)$ to write $\epsilon(v^2,\rho)$, and similarly for $v_0^2$. Restricting to Banach space of parameters $E$,  this can be formally described through the IFT for Banach spaces. 

\begin{theorem}\label{bound3}
Let  $\mathcal{C}=(\mathrm{Stellar}^{kl}(I),M,R)$  be a cluster of stellar systems of degree $(k,l)$, with $M$ piecewise $C^2$, endowed with a fully MLI state function $f$ and with a mass-radius function $g$ fulfilling a continuity equation. Any locally monotone invertible causal condition on $f$ and any upper bound on the radius induce an upper bound on $f$ depending only on the additional parameters of $g$. 
\end{theorem}
\begin{proof}
Once more we start by writing $\rho (R)=\partial _R M(R,\delta)$. The difference is that, instead of using bounds on the right-hand side to get bounds on $f$, we consider the speed of sound $v^2(\rho,\epsilon)$, which becomes $v^2(R,\delta,\epsilon)$. If we have a bound for $v^2$ we can write $v^2(R,\delta,\epsilon)\leq v_0^2(R,\delta,\epsilon)$, translating to a bound on $\epsilon$ via the invertibility hypothesis. If $R\leq R_0$ is a bound on the radius, from the monotone property of $v_0^2$ we obtain $\epsilon (R,v^2\delta)\leq \epsilon_0(\delta)$, where $\epsilon_0(\delta)=\epsilon (R_0,v_0^2\delta)$. This gives the desired bound on $f$, due to the MLI hypothesis.  
\end{proof}

\section{Applications}\label{applications}

Here, we apply Theorem \ref{bound3} to stellar systems inspired in experimental mass-radius relations. By this we mean relations found in the astrophysical literature as optimal approximations to experimental data describing zero age main sequence (ZAMS) stars  and terminal age main sequence (TAMS) stars. These stars satisfy polytropic equations of state and the causal condition, so Theorem \ref{bound3} can be applied.

\subsection{Monomial mass-radius relations}\label{monomial}
A simple model for mass-radius relations of ZAMS and TAMS stars was developed in \cite{RMR_polinomial_1, RMR_polinomial_2} by considering  monomial polynomials. The model was also applied recently to neutron stars in \cite{RMR_polinomial_3}). This means that in the cluster $\mathcal{C}=(\mathrm{Stellar}^{kl}(I),M,R)$ of ZAMS and TAMS stars we have a mass-radius relation given by a function
\begin{equation}\label{RMR_polinomial}
g(M,R,a,b) = M-aR^b,
\end{equation}
where $a,b$ are real functions on $\mathrm{Stellar}^{kl}(I)$, just as $R,M$. By definition, a stellar system belongs to the main sequence (MS) when the temperature at its nucleus is high enough to enable hydrogen fusion in such a way that the system becomes stable. This generally happens if the mass is at least $0.1\,M_{\odot}$ \cite{stellar_structure}, so these systems are naturally endowed with a lower bound on the mass. On the other hand, the mass of a MS-star determines many of its properties, such as luminosity and MS lifetime. In other words, different classes of MS-stars are characterized by upper bounds $M\leq M_0$. 

Additionally, stars at the beginning of the MS have smaller radius than those near MS's end. Therefore, ZAMS and TAMS have intrinsic bounds $R\leq R_0$. Furthermore, each partial derivative of $g$ in (\ref{RMR_polinomial}) does not vanish except at $R=0$ and $b=0$; the point $R=0$ is excluded by the bound $M\geq 0.1\,M_{\odot}$, and $b=0$ is excluded from experimental data. Therefore, Theorem \ref{bound} applies, giving the following corollary:
\begin{corollary}\label{corolario1}
Let $\mathcal{C}=(\mathrm{Stellar}^{kl}(I),M,R)$ be a cluster of ZAMS or TAMS stars fulfilling natural upper bounds $M\leq M_0$ and $R\leq R_0$. Then any function of state of $\mathcal{C}$ depending monotonically on $a$ and $b$ is bounded from above. 
\end{corollary}
Because ZAMS and TAMS stars are MS-stars, they follow Eddington's standard model, which means that polytropic state functions $f=p-K\rho^{\gamma}$ are good models to be chosen. In order to use Corollary \ref{corolario1} on $f$, we need some dependence on $a$ and $b$. If the dependence is on $K$, i.e, if $K(a,b)$ is a monotone function, Corollary \ref{corolario1} leads to bounds on $\gamma$ in terms of $b$, $\rho(R_0)$ and $P(R_0)$. If the dependence is on $\gamma$, we find bounds on $K$ in terms of the same parameters. Finally, if both $K$ and $\gamma$ depend of $a,b$, we get bounds on $b$ in terms of $\rho(R_0)$ and $P(R_0)$.

In Eddington's standard model, the MS-stars fulfill the continuity equation (\ref{tov2}), allowing us to apply Theorem \ref{bound2} and find bounds even when  $K$ and $\gamma$ do not depend on $a,b$. Explicitly, Equation (\ref{tov2}) combined with the mass-radius function (\cite{RMR_polinomial_1}) allows us to write the density as
\begin{equation}\label{RMR_polinomial_densidade}
\rho(R,a,b) = \dfrac{abR^{b-3}}{4\pi}
\end{equation}
The speed of the sound within the star becomes
\begin{equation}\label{speed_sound_RMR}
v(R,a,b) = k\left(\frac{ab}{4\pi}\right)^\beta R^{\beta-1} \; \text{with}\; \beta = \gamma(b-3)
\end{equation}
If we assume bounds on $M'$ we get bounds on (\ref{RMR_polinomial_densidade}) and, therefore, on the parameters of the polytropic state equation, as well as on the speed of sound (\ref{speed_sound_RMR}), as ensured by Theorem \ref{bound2}. For instance, if $M'\leq M_0$ and we are working only with stars of radius $R\leq R_0$ we find that $\gamma$ must satisfy $M_0^{1/\gamma} \geq abR_0^{b-1}$. This can also be understood as an upper bound on the radius that a polytropic MS-star of mass $M_0$ fulfilling the mass-radius relation (\ref{RMR_polinomial_densidade}) may have: $R_0 \leq (\frac{M_0 ^{1/\gamma}}{ab})^{1/(b-1)}$. Just to illustrate, for a massive ZAMS stars, say with $M_0 \approx 120\,M_{\odot}$, we have \cite{RMR_polinomial_1} $\gamma =3$, $a=0.85$ and $b=0.67$, so that $b-1<0$, implying that $R_0$ very small, agreeing with massive MS-stars only staying a short time as ZAMS stars. On the other hand, $b=1.78$ for $M_0 \approx 120\,M_{\odot}$ TAMS stars, meaning that in the terminal stage massive MS-stars may have a large radius.

We could also go in the direction of Theorem \ref{bound3} and use causal conditions (instead of conditions on $M'$) to get bounds. In natural units, the canonical choice of causal condition is $v< 1$. Assuming it and working in the same regime $R\leq R_0$,
\begin{equation}
k < \left[\left(\frac{ab}{4\pi}\right)^\beta R_0^{\beta-1}\right]^{-1} \; \text{with}\; \beta = \gamma(b-3).
\end{equation}

\subsection{Rational function mass-radius relation}
In Subsection \ref{monomial}, we studied monomial mass-radius relations for ZAMS and TAMS stars, found in \cite{RMR_polinomial_1,RMR_polinomial_2} following analysis of experimental data. However, for ZAMS stars with luminosity $Z = 0.02$, the monomial relation can be replaced by a rational one, as pointed in \cite{RMR_racional_1, RMR_racional_2}. So, let us consider mass-radius functions
\begin{equation}\label{RMR_racional}
g(M,R,a_i,c_i,b_i,d_i) = p(M,a_i,c_i) - Rq(M,b_i,d_i),
\end{equation}
on a cluster of stellar systems, where $p(M) =\sum_i a_iM^{b_i}$ and $q(M) = \sum_j c_jM^{d_j}$ are polynomials and $a_i, b_i,c_i,d_i$ are real functions on the cluster. We can apply Theorems \ref{bound}, \ref{bound2} and \ref{bound3} to find constraints on the possible state functions. Since we are interested in MS-stars, we study the distinguished polytropic state function. Let us see how Theorems \ref{bound2} and \ref{bound3} works in that case. Notice that $g(M,R,a_i,c_i,b_i,d_i)=0$ can be globally solved for $R$ as
\begin{equation}
R(M,a_i,c_i,b_i,d_i) = \frac{p(M,a_i,c_i)}{q(M,b_i,d_i)}.
\end{equation}
This function is clearly piecewise $C^1$ and, for each fixed $A\equiv (a_i,b_i,c_i,d_i)$, it is singular in a finite number of points. In the neighborhood of each regular point, $R(M,a_i,b_i,c_i,d_i)$ can be inverted for $M$, so that we have $M(R,A)$. Using the continuity equation (\ref{tov2}), we can write $\rho(R,A,M') = M'/4\pi R^2$. Setting bounds $M'\leq M_0$ and working with $R=R_0$, we get a bound $\rho \leq M^0/4pi R_0^2$ and, therefore, a bound on the polytropic function of state, as in Theorem \ref{bound2}. If instead of $M'\leq M_0$ we impose the canonical causal condition $v^2=k\gamma \rho^{\gamma -1}<1$, recalling that $\rho(R,A,M'(R,A))$, we get the constraint 
\begin{equation}
k<\frac{1}{\gamma \rho^{\gamma-1}_0(A)},\; \text{where}\; \rho_0(A)=\rho(R_0,A),
\end{equation}
exactly as in Theorem \ref{bound3}.

\section{Concluding Remarks}
We established, axiomatically and in great generality, that clusters with finite mass and finite radius can only accommodate mass-radius relations and equations of state simultaneously when certain bounds are satisfied.

Mass-radius relations can be found experimentally, while equations of state arise from theoretical modeling. In this context, we conclude that experimental data constrains \emph{a priori} the possible theoretical models. We believe that these general results point towards an axiomatic formulation of Astrophysics, a problem pointed and extensively worked by Chandrasekhar.     
\addtolength{\textheight}{-12cm}   

\onecolumngrid


\begin{thebibliography}{}
\bibitem{weinberg} 
    Weinberg, S., \emph{Gravitation and Cosmology: Principles and Applications of the General Theory of Relativity}, John Wiley \& Sons, 1972.
\bibitem{stellar_structure} 
    Chandrasekhar, S., \emph{An Introduction to the Study of Stellar Structure}, Dover, 1967.
\bibitem{generalized_norm} 
    Beer, G., Vanderwerff, J. Set-Valued Var. Anal vol.23 (2015) 613-630. 
\bibitem{IFT_cartan} 
    Cartan, H., \emph{Calcul Différentiel}, Hermann, 1967. 
\bibitem{Nash_Moser} 
    Hamilton, R. S., Bull. Amer. Math. Soc. (N.S.) vol. 7, n. 1 (1982), 65-222.
\bibitem{IFT_geral} 
    Glöckner, H., Isr. J. Math. (2006) 155-205.
\bibitem{IFT_loc_convex} 
    Iséki, K., Proc. Japan Acad. vol. 41, n.2 (1965), 147-149; Yamamuro, S., Bull. Austral. Math. Soc. vol. 12 (1975), 183-209; Lorenz, K. L., Math. Nach. Vol. 129, i.1, (1986) 91-101
\bibitem{IFT_adequado} 
    Hiltunen, S., Studia Mathematica 134 (1999), 235-250.  
\bibitem{biblia_politropes} Horedt, G. P., Polytropes: Applications in Astrophysics and Related Fields, Kluwer Academic Publisher, 2004.
\bibitem{eddington} Eddington, A. S., \emph{The Internal Constitution of Stars},  Cambridge University Press, 1926.
\bibitem{RMR_polinomial_1} 
    Demircan, O. \& Kahraman, G. Astrophys Space Sci (1991) 181: 313.
\bibitem{RMR_polinomial_2} 
    Giménez, A. \& Zamorano, J. Astrophys Space Sci (1985) 114: 259. 
\bibitem{RMR_polinomial_3}
    Zhen-Yu Zhu, En-Ping Zhou, Ang Li, ApJ 862,98 (2018)
\bibitem{RMR_racional_1} 
    Tout, C. A., Pols, O. R., Eggleton, P. P., \& Han, Z. Monthly Notices of the Royal Astronomical Society, Volume 281, Issue 1, pp. 257-262.
\bibitem{RMR_racional_2}
    Hurley, J. R., Pols, O. R., \& Tout, C. A. Monthly Notices of the Royal Astronomical Society, Volume 315, Issue 3, pp. 543-569.
\end{thebibliography}
\end{document}